\documentclass[11pt,a4paper]{article}

\usepackage[utf8]{inputenc}
\usepackage[T1]{fontenc}
\usepackage{lmodern}

\usepackage{amsmath,amssymb,amsthm}
\usepackage{bm}
\usepackage{mathtools}

\usepackage{booktabs}
\usepackage{array}
\usepackage{float}          %

\usepackage{tikz}
\usepackage{pgfplots}
\pgfplotsset{compat=1.18}
\usetikzlibrary{arrows.meta, patterns, calc, decorations.markings, decorations.pathreplacing, positioning}
\usepgfplotslibrary{fillbetween}

\usepackage[margin=2.5cm]{geometry}
\usepackage{setspace}
\usepackage{parskip}

\usepackage[authoryear,round]{natbib}

\usepackage[colorlinks=true,linkcolor=blue,citecolor=blue,urlcolor=blue]{hyperref}
\usepackage{cleveref}

\newtheorem{theorem}{Theorem}
\newtheorem{proposition}{Proposition}
\newtheorem{conjecture}{Conjecture}
\newtheorem{corollary}{Corollary}

\theoremstyle{definition}
\newtheorem{definition}{Definition}
\newtheorem{example}{Example}
\theoremstyle{remark}
\newtheorem*{remark}{Remark}

\newcommand{\dd}{\mathrm{d}}

\newcommand{\tw}{\tau_w}

\newcommand{\V}{\mathbf{V}}
\newcommand{\one}{\mathbf{1}}

\title{Spectral Portfolio Theory: From SGD Weight Matrices to Wealth Dynamics}
\author{Anders G Fr{\o}seth\thanks{Independent Researcher.
  E-mail: \href{mailto:indrefjorden@pm.me}{indrefjorden@pm.me}.}}
\date{\today}

\begin{document}
\maketitle

\begin{abstract}
We develop spectral portfolio theory by establishing a direct identification: neural network weight matrices trained on stochastic processes are portfolio allocation matrices, and their spectral structure encodes factor decompositions and wealth concentration patterns.
The three forces governing stochastic gradient descent (SGD)---gradient signal, dimensional regularisation, and eigenvalue repulsion---translate directly into portfolio dynamics: smart money, survival constraint, and endogenous diversification.
The spectral properties of SGD weight matrices transition from Marchenko--Pastur statistics (additive regime, short horizon) to inverse-Wishart via the free log-normal (multiplicative regime, long horizon), mirroring the transition from daily returns to long-run wealth compounding.
We unify the cross-sectional wealth dynamics of \citet{BouchaudMezard2000}, the within-portfolio dynamics of \citet{OlsenEtAl2025}, and the scalar Fokker--Planck framework via a common spectral foundation.
A central result is the Spectral Invariance Theorem: any isotropic perturbation to the portfolio objective preserves the singular-value distribution up to scale and shift, while anisotropic perturbations produce spectral distortion proportional to their cross-asset variance.
We develop applications to portfolio design, wealth inequality measurement, tax policy, and neural network diagnostics.
In the tax context, the invariance result recovers and generalises the neutrality conditions of \citet{Froeseth2026N}.
\end{abstract}

\section{Introduction}\label{sec:intro}

Consider a feedforward neural network trained to learn a stochastic process---a drift function, a transition density, or a score function---from observed trajectories.
At each layer, the weight matrix $W \in \mathbb{R}^{m \times n}$ evolves under stochastic gradient descent.

This paper takes the following identification as its foundation: the weight matrix at a given layer is a portfolio allocation matrix.
Row $i$ specifies how capital is distributed across $n$ assets in state $i$, and columns index assets.
Training the network on data from a stochastic process is equivalent to running an adaptive portfolio optimizer on that process.

The spectral theory of SGD dynamics developed by \citet{OlsenEtAl2025} reveals that this identification is not merely an analogy: the stationary portfolio structure, the factor decomposition, and the concentration--diversification tradeoff emerge as direct mathematical consequences.
This paper develops the framework systematically, proving a spectral invariance theorem for isotropic perturbations and applying the results to portfolio design, wealth inequality, tax policy, and neural network diagnostics.

The main contributions are as follows.
First, we establish that the three forces in the singular-value evolution equation have direct portfolio interpretations: gradient signal encodes smart money (return-seeking), dimensional regularisation captures an endogenous survival constraint, and eigenvalue repulsion implements endogenous diversification without explicit constraints.
Second, we characterize the stationary spectral distribution as having a gamma-type bulk with power-law tail, giving rise to the core--satellite portfolio structure observed empirically in institutional wealth and household portfolio data.
Third, we show that the spectral transition from additive (Marchenko--Pastur) to multiplicative (inverse-Wishart) regimes is governed by the free log-normal distribution and the matrix Kesten problem, establishing a precise mathematical duality between short-horizon and long-horizon portfolio dynamics.
Fourth, we unify the cross-sectional wealth model of Bouchaud and Mézard, the within-portfolio model of Olsen et al., and the scalar Fokker--Planck framework via spectral decomposition.
Fifth, we prove the Spectral Invariance Theorem: any isotropic perturbation to the portfolio objective preserves the spectral shape of the allocation matrix, while anisotropic perturbations produce distortion proportional to their cross-asset variance.
Sixth, we develop applications to portfolio design, wealth inequality measurement, tax policy, and neural network diagnostics, demonstrating the breadth of the framework.

The paper is organized in five parts.
Part~I (Sections~\ref{sec:setup}--\ref{sec:learning}) establishes the foundations: the learning setup (Figure~\ref{fig:learning-setup}), the weight matrix--portfolio identification, the SVD decomposition into eigenportfolios, the three forces of SGD, the stationary spectral distribution, and factor complexity.
Part~II (Sections~\ref{sec:regimes}--\ref{sec:fp-connection}) develops the dynamics: timescale regimes governing the additive-to-multiplicative transition, the ergodicity gap as a spectral observable, the aggregation from matrix to scalar via Itô projection, and the Fokker--Planck connection.
Part~III (Sections~\ref{sec:bouchaud}--\ref{sec:loss}) presents general results: the Bouchaud--Mézard unification, the Spectral Invariance Theorem and its anisotropic counterpart, the Spectral Invariance Conjecture, and the loss--utility correspondence.
Part~IV (Section~\ref{sec:applications}) develops applications to portfolio design, wealth inequality, tax policy, and neural network diagnostics.
Part~V (Sections~\ref{sec:predictions}--\ref{sec:conclusion}) collects testable predictions, open questions, literature context, and conclusions.

\section{Setup: Weight Matrices as Allocation Matrices}\label{sec:setup}

The starting point is a stochastic process observed through discrete trajectories.
The process follows a stochastic differential equation
\begin{equation}\label{eq:data-sde}
  dx = v(x)\,dt + \sqrt{2D}\,dB,
\end{equation}
where $v(x)$ is the drift function and $D$ the diffusion constant.
A neural network is trained to learn either the drift $v(x)$, the transition density $p(x_{t+\Delta t} \mid x_t)$, or the score function $\nabla_x \log p(x)$ from observed trajectory data $\{x_0, x_{\Delta t}, x_{2\Delta t}, \ldots\}$.
The weight matrix $W$ at each layer is updated by stochastic gradient descent on a loss function $\mathcal{L}(W)$ that is quadratic in the prediction residual (see Section~\ref{sec:loss} for details).

\begin{figure}[H]
\centering
\begin{tikzpicture}[
    >=Stealth,
    lbl/.style={font=\scriptsize, text=black!60},
    arrow/.style={->, thick, shorten >=2pt, shorten <=2pt},
  ]

  \node[font=\small\bfseries] at (-4.5, 3.2) {Data: trajectories};

  \begin{scope}[xshift=-4.5cm, yshift=0cm]
    \draw[->, thick, black!50] (-1.8, -0.3) -- (1.8, -0.3) node[right, lbl] {$t$};
    \draw[->, thick, black!50] (-1.8, -0.3) -- (-1.8, 2.8) node[above, lbl] {$x(t)$};

    \draw[thick, blue!60, smooth] plot coordinates
      {(-1.5,0.0) (-1.0,0.3) (-0.5,0.6) (0.0,0.5) (0.5,1.0) (1.0,1.4) (1.5,1.6)};
    \draw[thick, red!50, smooth] plot coordinates
      {(-1.5,0.2) (-1.0,0.6) (-0.5,1.1) (0.0,1.3) (0.5,1.2) (1.0,1.8) (1.5,2.3)};
    \draw[thick, green!50!black, smooth] plot coordinates
      {(-1.5,0.1) (-1.0,0.0) (-0.5,0.3) (0.0,0.8) (0.5,0.7) (1.0,0.6) (1.5,1.0)};

    \foreach \x/\yb/\yr/\yg in {-1.5/0.0/0.2/0.1, -0.5/0.6/1.1/0.3, 0.5/1.0/1.2/0.7, 1.5/1.6/2.3/1.0} {
      \fill[blue!60] (\x, \yb) circle (1.5pt);
      \fill[red!50] (\x, \yr) circle (1.5pt);
      \fill[green!50!black] (\x, \yg) circle (1.5pt);
    }

    \node[font=\scriptsize, text=black!70, align=center] at (0, -0.9)
      {$dx = v(x)\,dt + \sqrt{2D}\,dB$};
  \end{scope}

  \draw[arrow, thick, black!50] (-2.2, 1.2) -- (-0.8, 1.2)
    node[midway, above=2pt, lbl] {$\{x_t, x_{t+\Delta t}\}$};

  \node[font=\small\bfseries] at (1.2, 3.2) {Neural network};

  \begin{scope}[xshift=1.2cm, yshift=1.2cm]
    \node[draw, circle, minimum size=0.45cm, thick, fill=gray!10] (i1) at (-1.2, 0.6) {};
    \node[draw, circle, minimum size=0.45cm, thick, fill=gray!10] (i2) at (-1.2, -0.6) {};
    \node[lbl] at (-1.2, -1.15) {$x_t$};

    \node[draw, circle, minimum size=0.45cm, thick, fill=blue!10] (h1) at (0, 1.0) {};
    \node[draw, circle, minimum size=0.45cm, thick, fill=blue!10] (h2) at (0, 0.2) {};
    \node[draw, circle, minimum size=0.45cm, thick, fill=blue!10] (h3) at (0, -0.6) {};

    \node[draw, rounded corners=2pt, fill=blue!6, font=\scriptsize, thick,
          minimum width=0.9cm, minimum height=1.8cm] at (0, 0.2) {};
    \node[font=\tiny, blue!70] at (0, 0.2) {$W$};

    \node[draw, circle, minimum size=0.45cm, thick, fill=orange!15] (o1) at (1.2, 0.0) {};
    \node[lbl] at (1.2, -1.15) {$\hat{v}(x_t; W)$};

    \foreach \i in {1,2} {
      \foreach \h in {1,2,3} {
        \draw[black!30, thin] (i\i) -- (h\h);
      }
    }
    \foreach \h in {1,2,3} {
      \draw[black!30, thin] (h\h) -- (o1);
    }
  \end{scope}

  \draw[arrow, thick, black!50] (2.9, 1.2) -- (4.1, 1.2)
    node[midway, above=2pt, lbl] {prediction};

  \node[font=\small\bfseries] at (5.8, 3.2) {Loss function};

  \begin{scope}[xshift=5.8cm, yshift=1.2cm]
    \node[draw, rounded corners=5pt, thick, fill=yellow!8, font=\scriptsize,
          text width=3.0cm, align=center, minimum height=2.6cm] at (0, 0)
      {\textbf{MLE:}\\[3pt]
       $\displaystyle\frac{1}{T}\sum_t\Bigl\|\frac{\Delta x_t}{\Delta t} - \hat{v}\Bigr\|^2$\\[8pt]
       \textbf{Score matching:}\\[3pt]
       $\frac{1}{2}\|s_\theta - \nabla\!\log p\|^2$};
  \end{scope}

  \node[lbl, text width=3.0cm, align=center] at (5.8, -0.45)
    {both quadratic in residual};

  \node[font=\scriptsize, text=black!50] (sgdlabel) at (3.5, -1.15)
    {SGD: $\;W \leftarrow W - \eta\,\nabla_W\mathcal{L} + \text{noise}$};
  \draw[->, thick, black!40, rounded corners=3pt]
    (5.8, -1.15) -- (5.8, -1.7) -- (1.2, -1.7) -- (1.2, -0.6);

\end{tikzpicture}
\caption{The learning setup.
A stochastic process~\eqref{eq:data-sde} generates trajectory data (left).
A neural network with weight matrix $W$ at each layer learns the drift function $\hat{v}(x; W)$ or the score function $s_\theta(x)$ (centre).
The loss function is quadratic in the prediction residual (right), whether using maximum likelihood or score matching.
SGD updates $W$ iteratively, and the stationary spectral distribution of $W$ encodes the structure of the learned process.}
\label{fig:learning-setup}
\end{figure}
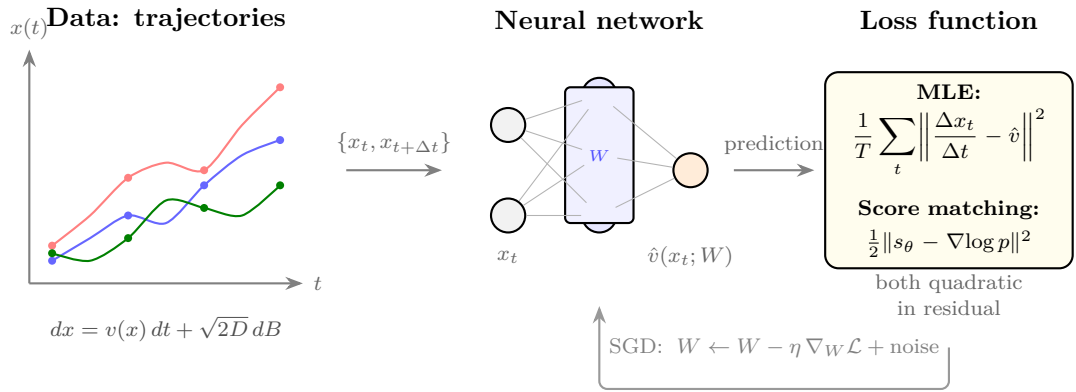

In \citet{OlsenEtAl2025}, the weight matrix $W \in \mathbb{R}^{m \times n}$ evolves under the stochastic differential equation
\begin{equation}\label{eq:olsen-sde}
  dW = -\eta \frac{\partial \mathcal{L}}{\partial W}\,dt + \sqrt{2\eta D}\,d\mathcal{W},
\end{equation}
where $\eta$ is the learning rate, $D$ an effective diffusion constant, and $d\mathcal{W}$ is matrix-valued Wiener noise.

We relabel the dimensions: set $m = T$ (time periods or states of the world) and $n = N$ (number of assets).
Then $W_{ti}$ represents the portfolio weight on asset $i$ in state $t$.
The allocation matrix $W \in \mathbb{R}^{T \times N}$ evolves as investors rebalance in response to return signals, information noise, and constraints implicit in the learning dynamics.

The singular value decomposition $W = U\Sigma V^\top$ decomposes this allocation into principal components.
The left singular vectors $u_k \in \mathbb{R}^T$ show when each factor is relevant (temporal patterns).
The right singular vectors $v_k \in \mathbb{R}^N$ show which assets compose each factor (eigenportfolios).
The singular values $\sigma_k$ measure the magnitude of allocation to each eigenportfolio.

This decomposition reveals the factor structure naturally: large singular values correspond to important factors (market-wide moves, systematic risks); small singular values correspond to idiosyncratic or second-order effects.
The right singular vectors $v_k$ define eigenportfolio compositions, and the left singular vectors $u_k$ indicate when each factor dominates.
Figure~\ref{fig:identification} illustrates this relabelling.

\begin{figure}[H]
\centering
\begin{tikzpicture}[
    >=Stealth,
    node distance=1.2cm,
    matbox/.style={draw, rounded corners=3pt, minimum width=2.6cm, minimum height=1.2cm,
                   font=\small, align=center, thick, fill=blue!6},
    svdbox/.style={draw, rounded corners=3pt, minimum width=1.3cm, minimum height=0.7cm,
                   font=\small, align=center, thick},
    lbl/.style={font=\scriptsize, text=black!60},
    arrow/.style={->, thick, shorten >=2pt, shorten <=2pt},
  ]

  \node[lbl, font=\scriptsize\bfseries] at (-4.8, 2.8) {Neural Network Layer};

  \foreach \i/\y in {1/2.0, 2/1.4, 3/0.8} {
    \node[draw, circle, minimum size=0.5cm, thick, fill=gray!10] (in\i) at (-4.8, \y) {};
  }
  \node[lbl] at (-4.8, 0.3) {$x_t \in \mathbb{R}^m$};

  \node[matbox] (W) at (-2.0, 1.4) {$W \in \mathbb{R}^{m \times n}$\\[1pt]\scriptsize weight matrix};

  \foreach \i/\y in {1/2.3, 2/1.7, 3/1.1, 4/0.5} {
    \node[draw, circle, minimum size=0.5cm, thick, fill=gray!10] (out\i) at (0.8, \y) {};
  }
  \node[lbl] at (0.8, 0.0) {$y_t = W^\top\! x_t$};

  \foreach \i in {1,2,3} {
    \draw[arrow, black!40] (in\i) -- (W.west |- in\i);
  }
  \foreach \i in {1,2,3,4} {
    \draw[arrow, black!40] (W.east |- out\i) -- (out\i);
  }

  \draw[<->, very thick, blue!50] (-2.0, -0.1) -- (-2.0, -0.9)
    node[midway, right=3pt, font=\footnotesize\bfseries, blue!60] {$\equiv$};

  \node[lbl, font=\scriptsize\bfseries] at (-4.8, -1.2) {Portfolio Allocation};

  \foreach \i/\y/\lab in {1/-1.7/t{=}1, 2/-2.3/t{=}2, 3/-2.9/t{=}T} {
    \node[draw, rounded corners=2pt, minimum width=0.7cm, minimum height=0.4cm,
          thick, fill=orange!10, font=\scriptsize] (st\i) at (-4.8, \y) {$\lab$};
  }
  \node[lbl] at (-4.8, -3.35) {states};

  \node[matbox, fill=orange!6] (A) at (-2.0, -2.3) {$W_{ti}$\\[1pt]\scriptsize allocation matrix};

  \foreach \i/\y/\lab in {1/-1.4/{i{=}1}, 2/-2.0/{i{=}2}, 3/-2.6/{i{=}3}, 4/-3.2/{i{=}N}} {
    \node[draw, rounded corners=2pt, minimum width=0.7cm, minimum height=0.4cm,
          thick, fill=orange!10, font=\scriptsize] (as\i) at (0.8, \y) {$\lab$};
  }
  \node[lbl] at (0.8, -3.65) {assets};

  \foreach \i in {1,2,3} {
    \draw[arrow, black!40] (st\i) -- (A.west |- st\i);
  }
  \foreach \i in {1,2,3,4} {
    \draw[arrow, black!40] (A.east |- as\i) -- (as\i);
  }

  \node[lbl, font=\scriptsize\bfseries] at (4.8, 2.8) {SVD: $W = U\Sigma V^\top$};

  \node[svdbox, fill=green!8] (U) at (3.2, 1.8) {$U$};
  \node[svdbox, fill=yellow!12] (S) at (4.8, 1.8) {$\Sigma$};
  \node[svdbox, fill=red!8] (V) at (6.4, 1.8) {$V^\top$};

  \draw[arrow, thick, blue!50] (1.7, 1.4) -- (U.west);

  \node[lbl, text width=2.2cm, align=center] at (3.2, 0.75) {$u_k \in \mathbb{R}^T$\\[1pt]temporal\\patterns};
  \node[lbl, text width=1.8cm, align=center] at (4.8, 0.75) {$\sigma_k$\\[1pt]factor\\magnitudes};
  \node[lbl, text width=2.2cm, align=center] at (6.4, 0.75) {$v_k \in \mathbb{R}^N$\\[1pt]eigenportfolio\\compositions};

  \draw[densely dotted, black!25] (U.south) -- (3.2, 1.25);
  \draw[densely dotted, black!25] (S.south) -- (4.8, 1.25);
  \draw[densely dotted, black!25] (V.south) -- (6.4, 1.25);

  \node[draw, rounded corners=3pt, thick, fill=gray!5, font=\scriptsize,
        text width=3.2cm, align=center] (sgd) at (4.8, -1.8)
    {SGD dynamics:\\[2pt]
     $dW = -\eta\nabla\mathcal{L}\,dt + \sqrt{2\eta D}\,d\mathcal{W}$\\[4pt]
     $\equiv$ adaptive portfolio\\rebalancing};

  \draw[arrow, thick, black!40] (A.east) -- (sgd.west);

  \draw[arrow, thick, black!40] (sgd.north) -- (4.8, -0.1)
    node[midway, right=2pt, lbl] {spectra};

\end{tikzpicture}
\caption{The neural network--portfolio identification.
A single layer with weight matrix $W \in \mathbb{R}^{m \times n}$ maps input $x_t$ to output $y_t = W^\top x_t$ (top left).
Relabelling rows as states and columns as assets gives the allocation matrix $W_{ti}$ (bottom left).
The SVD decomposes $W$ into temporal patterns $u_k$, factor magnitudes $\sigma_k$, and eigenportfolio compositions $v_k$ (top right).
SGD dynamics on $W$ are equivalent to adaptive portfolio rebalancing (bottom right); the stationary spectral density of $\sigma_k$ is determined by the signal-to-noise ratio $\beta_1/\eta D$.}
\label{fig:identification}
\end{figure}
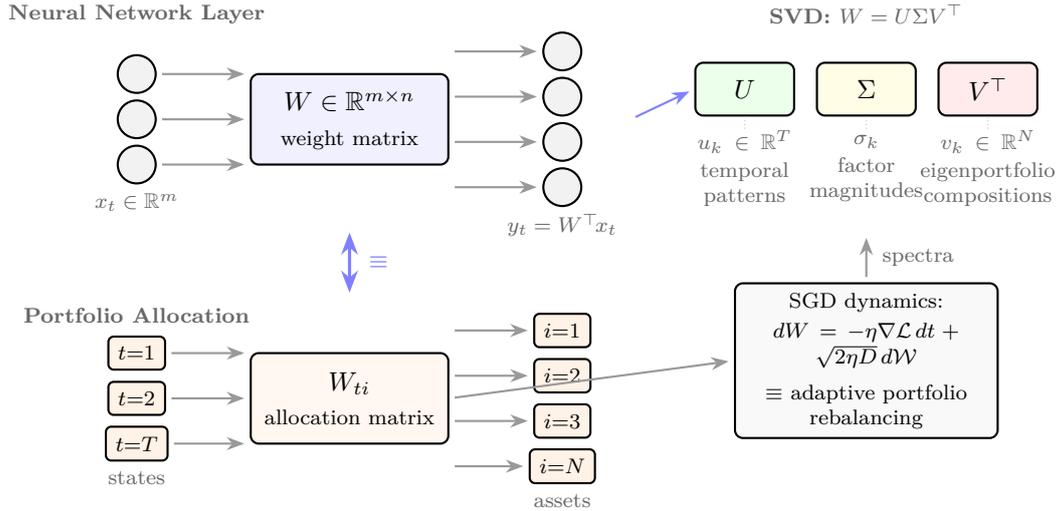

\section{Portfolio Dynamics: The Three Forces of SGD}\label{sec:forces}

By \citet{OlsenEtAl2025} (Theorem 3.1), the singular values evolve as
\begin{equation}\label{eq:sv-sde}
  d\sigma_k = \Bigl[
    -\eta\, u_k^\top(\nabla_W \mathcal{L})\,v_k
    + \eta D\Bigl(\frac{m-n+1}{2\sigma_k}
    + \sum_{j \neq k}\frac{\sigma_k}{\sigma_k^2 - \sigma_j^2}\Bigr)
  \Bigr]dt
  + \sqrt{2\eta D}\,d\beta_k.
\end{equation}

This equation contains three distinct forces, each with a clear portfolio interpretation.

\subsection{Gradient Signal: Smart Money}\label{subsec:signal}

The term
\begin{equation}
  -\eta\, u_k^\top(\nabla_W \mathcal{L})\,v_k
\end{equation}
is the projected return signal pushing capital toward high-performing eigenportfolios.
This is the optimizer's best estimate of where marginal utility of wealth is highest.
In a mean--variance framework with loss $\mathcal{L} = -Wr + \tfrac{\gamma}{2}W^2$, this term becomes the excess return on factor $k$ adjusted for risk aversion.
The negative sign reflects gradient descent: the optimizer increases weights on factors with positive marginal returns.

This force implements \textbf{smart money dynamics}: capital flows toward factors with superior risk-adjusted returns.

\subsection{Dimensional Regularization: Survival Constraint}\label{subsec:survival}

The term
\begin{equation}
  \eta D\,\frac{m - n + 1}{2\sigma_k}
\end{equation}
is a $1/\sigma_k$ restoring force that prevents any eigenportfolio loading from reaching zero.
Small positions get pushed up.
Economically, this is a minimum-position constraint arising \emph{endogenously} from the noise structure.
An investor observing returns with noise cannot rationally set any factor exposure exactly to zero, because zero exposure means zero information---creating an explore--exploit tension.

The prefactor $m - n + 1$ is the aspect ratio of the allocation matrix, encoding how many degrees of freedom remain after fitting the primary factors.
More time periods (larger $m$) strengthen the survival constraint; more assets (larger $n$) weaken it.

This force implements \textbf{endogenous survival}: positions that would shrink to zero due to poor current performance are sustained by information and exploration incentives.

\subsection{Eigenvalue Repulsion: Endogenous Diversification}\label{subsec:repulsion}

The term
\begin{equation}
  \eta D \sum_{j \neq k}\frac{\sigma_k}{\sigma_k^2 - \sigma_j^2}
\end{equation}
prevents any two eigenportfolio loadings from becoming equal.
In portfolio language, this is \textbf{endogenous diversification}.
Even without an explicit diversification constraint, stochastic rebalancing naturally separates factor exposures.
When two factors carry similar weight ($\sigma_k \approx \sigma_j$), the denominator diverges, creating an unstable repulsion that forces differentiation.

This is a new micro-foundation for diversification: it arises not from risk aversion (which is in the gradient term), but from the \emph{noise structure of learning under partial information}.
It holds regardless of the loss function or risk aversion parameter.

\begin{remark}[Connection to Merton]
In the continuous-time portfolio problem, the optimal allocation involves the inverse covariance matrix $\Sigma^{-1}$ \citep{Merton1969,Merton1971}.
The eigenvalue repulsion ensures that the estimated covariance matrix remains well-conditioned---it is a regularisation of the Merton solution that emerges from the learning dynamics.
\end{remark}

\begin{table}[H]
\centering
\small
\begin{tabular}{@{}p{3.2cm}p{3.8cm}p{5.5cm}@{}}
\toprule
\textbf{This paper} & \textbf{Nearest existing concept} & \textbf{Key distinction} \\
\midrule
Smart money (gradient signal)
  & Informed investor flows \citep{Gruber1996}
  & Optimisation on available data, not asymmetric information; compatible with market efficiency \\[6pt]
Survival constraint (dim.\ regularisation)
  & Safety-first / Kelly criterion \citep{Roy1952}
  & Survival of information channels (factor exposures), not of the investor \\[6pt]
Endogenous diversification (eigenvalue repulsion)
  & Mean--variance diversification \citep{Markowitz1952}
  & Emerges from noise structure of learning; independent of risk aversion \\
\bottomrule
\end{tabular}
\caption{Terminology disambiguation.  The portfolio labels proposed in this paper echo but differ from established usage.  Each force in the singular-value SDE~\eqref{eq:sv-sde} has a nearest classical counterpart; the distinctions are summarised here.}
\label{tab:terminology}
\end{table}

\begin{remark}[Adaptive dynamics]\label{rem:adaptive}
Taken together, the three forces describe an explore--exploit--diversify triad: the gradient signal exploits current information, the survival constraint maintains exploratory positions, and eigenvalue repulsion enforces differentiation.
This is reminiscent of the Adaptive Markets Hypothesis \citep{Lo2004,Lo2017}, in which market behaviour is driven by evolutionary dynamics among strategies.
The spectral framework can be viewed as providing a mathematical instantiation of Lo's qualitative programme.
\end{remark}

\section{The Stationary Portfolio: Bulk and Tail}\label{sec:stationary}

By \citet{OlsenEtAl2025} (Theorem 3.2), the stationary distribution of singular values is
\begin{equation}\label{eq:stat-sv}
  p_\sigma(\sigma) \propto \sigma^{(m-n+1)/2}\,e^{-(\beta_1 / 4\eta D)\,\sigma^2},
\end{equation}
where $\beta_1$ is a mean-field restoring-force constant.

This distribution exhibits a clear structure: a concentrated bulk followed by a power-law tail.

\begin{definition}[Core--Satellite Portfolio]\label{def:core-satellite}
At stationarity, we partition the eigenportfolios into two groups:
\begin{itemize}
  \item \textbf{Bulk}: most eigenportfolio loadings cluster in a concentrated region.
    This is the diversified \textit{core}---broad market exposure spread across many factors.
  \item \textbf{Tail}: a few large $\sigma_k$ values form a power-law tail ($\sigma^{(m-n+1)/2}$ for large $\sigma$).
    These are concentrated \textit{satellite} positions---large bets on specific factors.
\end{itemize}
This core--satellite structure is observed empirically in institutional portfolios and, crucially, in the Norwegian wealth register data of \citet{Froeseth2026N,Froeseth2026E,Froeseth2026S}.
\end{definition}

The tail exponent $(m - n + 1)/2$ depends on the aspect ratio of the allocation matrix.
For a market with $N$ assets observed over $T$ periods, the exponent is $(T - N + 1)/2$.
More observation (larger $T$) fattens the tail (allows more concentration); more assets (larger $N$) thins it (forces diversification).
This is testable against Norwegian portfolio data.

\begin{definition}[Effective Spectral Rank]\label{def:eff-rank}
Define the \emph{effective spectral rank} as
\begin{equation}
  r_{\text{eff}} = \frac{\sum_k \sigma_k}{\max_k \sigma_k},
\end{equation}
the ratio of total allocation to the largest single factor exposure.
This measures the complexity of the portfolio: $r_{\text{eff}} = 1$ means one factor dominates; $r_{\text{eff}} = N$ means all factors have equal weight.
The stationary distribution predicts a specific value of $r_{\text{eff}}$ as a function of the aspect ratio and the signal-to-noise ratio.
\end{definition}

\section{Network Learning: Factor Complexity and Data Structure}\label{sec:learning}

When a weight matrix $W$ is trained on data from a stochastic process, it encodes a factor decomposition of that process.
The singular values tell us the relative importance of each factor.

Consider three cases:

\textbf{Case 1: Simple underlying structure.}
If the data is generated by a drift function with simple structure (mean-reverting, polynomial, Fourier basis), the trained weight matrix will have a small number of large singular values followed by a sharp drop.
The bulk is thin, the tail is fat---a concentrated factor model.

\textbf{Case 2: Complex or noisy data.}
If the data is noisy or has complex structure, the trained matrix will have a more gradual decay of singular values.
The bulk is thick, the tail is thinner---a more diversified representation.

\textbf{Case 3: Pareto-distributed data.}
The spectral complexity measure $r_{\text{eff}}$ directly reflects the intrinsic dimensionality of the process.
For Pareto-distributed data (as in wealth), the exponent of $r_{\text{eff}}$ should match the Pareto exponent of the data.

This means the weight spectrum is not an artifact of the learning process; it is a faithful encoding of the underlying data structure.
In deep networks with $L$ layers, each layer performs a hierarchical factor decomposition: early layers capture primary factors, intermediate layers combine them into secondary factors, and the final layer produces the output.
This is the fund-of-funds structure: not a single portfolio, but a hierarchy of portfolios.

\section{Timescale Regimes and Spectral Transitions}\label{sec:regimes}

The spectral portfolio framework must reconcile two well-established facts about asset dynamics that operate on different timescales.

\subsection{The Additive Regime: Short Horizon}\label{subsec:additive}

At timescales from minutes to days, asset returns are approximately additive: $r_t \approx \mu\,\delta t + \sigma\,\delta B_t$, with heavy tails but finite variance over short windows.
This is the regime of the classical stylized facts \citep{Cont2001, BouchaudPotters2003}: fat-tailed return distributions, volatility clustering, and approximately uncorrelated returns.
The covariance matrix $\Sigma = \mathrm{Cov}(r_t)$ has a Marchenko--Pastur bulk with a few spikes (market factor, sector factors), and random matrix theory cleaning \citep{LalouxEtAl1999, LedoitWolf2004} is effective precisely because the bulk is well described by random matrix universality.

This is also the regime where \citet{OlsenEtAl2025}'s spectral theory of SGD applies most directly: weight updates are small additive perturbations, the matrix-valued SDE is well approximated by Dyson Brownian motion, and the stationary spectral density follows a gamma-type distribution with power-law tails.  \citet{AartsEtAl2025} provide independent confirmation from the physics side, showing that the eigenvalue dynamics of SGD weight matrices are described by Dyson Brownian motion with a Coulomb gas stationary distribution whose variance scales as $\alpha/|\mathcal{B}|$ (learning rate over batch size).

\subsection{The Multiplicative Regime: Long Horizon}\label{subsec:multiplicative}

At timescales from months to years, compounding dominates.
Wealth dynamics are inherently multiplicative: $\dd W / W = \mu\,\dd t + \sigma\,\dd B_t$, and the natural coordinate is log-wealth $x = \ln W$, which converts multiplicative dynamics into an additive SDE with drift $v = \mu - \tfrac{1}{2}\sigma^2$ and diffusion $D = \sigma^2$.
This is the regime of \citet{Froeseth2026S}, where the Fokker--Planck equation governs the wealth distribution and the Pareto tail emerges from drift--diffusion balance.

Tax effects and portfolio allocation decisions operate on this longer timescale.
The Heston model \citep{Heston1993} and other stochastic-volatility extensions capture both regimes: they remain multiplicative while capturing the volatility clustering persisting from the short-horizon regime.

\subsection{The $q$-Transformation: Quantifying the Crossover}\label{subsec:q-transform}

\citet{BouchaudPotters2003} introduce a continuous parametrisation of the additive-to-multiplicative crossover.
Define the price process
\begin{equation}\label{eq:q-transform}
  x(T) = x_0\bigl(1 + q(T)\,\xi(T)\bigr)^{1/q(T)},
  \qquad q(T) \in [0,1],
\end{equation}
where $\xi(T)$ is the fundamental random variable (normalised return).
At $q = 1$ the process is fully additive: $x = x_0(1 + \xi)$.
As $q \to 0$, one recovers the multiplicative limit: $(1 + q\xi)^{1/q} \to e^{\xi}$, so $x = x_0\,e^{\xi}$.

The index $q(T)$ is a quantitative measure of the additivity of the process at horizon $T$.
Empirically, $q$ decreases monotonically from near~1 at intraday timescales to near~0 at horizons of months to years.
The crossover timescale $T_c$ at which $q \approx 1/2$ is on the order of \emph{months} for liquid equity markets.

The $q$-transformation also governs the skewness:
\begin{equation}\label{eq:q-skewness}
  s(T) = -3\,q(T)\,\sigma(T),
\end{equation}
so that skewness vanishes in the multiplicative limit ($q \to 0$) and is largest in the additive regime.
This provides a measurable diagnostic: the observed skewness at a given horizon identifies where the process sits in the crossover.

For the spectral portfolio framework, $q(T)$ determines which universality class governs the covariance spectrum at a given horizon.
When $q \approx 1$ (additive, short horizon), the spectrum is Marchenko--Pastur with isolated signal spikes.
When $q \approx 0$ (multiplicative, long horizon), the spectrum transitions toward the free log-normal/inverse-Wishart regime described next.

\subsection{The Transition: Marchenko--Pastur to Inverse-Wishart}\label{subsec:transition}

The passage from additive to multiplicative regime reflects not merely a change of resolution but a change in the dominant mechanism and, crucially, a change in the spectral universality class.

\citet{PottersBouchaud2021} develop this transition explicitly in free probability language.
The additive central limit theorem for free random matrices gives the Wigner semicircle (Marchenko--Pastur for rectangular matrices).
The \emph{multiplicative} central limit theorem gives the \emph{free log-normal}, with S-transform
\begin{equation}\label{eq:free-lognormal}
  S_{\mathrm{LN}}(t) = \mathrm{e}^{-a/2 - bt},
\end{equation}
where $a$ controls the mean log-eigenvalue and $b$ the variance.
The free log-normal is stable under free products, as the semicircle is stable under free sums.
For small $a$, the density is indistinguishable from a Wigner semicircle; as $a$ grows, it develops the characteristic asymmetric shape of multiplicative dynamics.

\textbf{Multiplicative Dyson Brownian Motion.}
The additive DBM governing eigenvalue dynamics in the short-horizon regime has a multiplicative counterpart \citep{PottersBouchaud2021}:
\begin{equation}\label{eq:mult-dbm}
  \frac{\dd\lambda_i}{\dd t}
    = \frac{a}{2}\,\lambda_i
    + \frac{b}{N}\sum_{j \neq i}\frac{\lambda_i\lambda_j}{\lambda_i - \lambda_j}
    + \sqrt{\frac{b}{N}}\,\lambda_i\,\xi_i,
\end{equation}
where $\xi_i$ are independent Langevin noise terms.
Every term is proportional to $\lambda_i$: the drift, the repulsion, and the noise are all multiplicative.
At arbitrary time $t$, the spectral density is the free log-normal with parameters $ta$ and $tb$.

\textbf{Matrix Kesten Problem.}
The scalar Kesten recursion $Z_{n+1} = Z_n(1 + \zeta_n)$, which generates Pareto tails in \citet{BouchaudMezard2000,Froeseth2026S}, has a matrix generalization.
The corresponding Fokker--Planck equation is
\begin{equation}\label{eq:matrix-kesten-fp}
  \frac{\partial P}{\partial t}
    = -\frac{\partial}{\partial U}\bigl[(1 + mU)P\bigr]
    + \frac{\sigma^2}{2}\frac{\partial^2}{\partial U^2}\bigl[U^2 P\bigr],
\end{equation}
with stationary distribution $P_{\mathrm{eq}}(U) \propto U^{-1-\mu}\,\mathrm{e}^{-2/(\sigma^2 U)}$, an inverse-gamma with power-law tail exponent
\begin{equation}\label{eq:matrix-kesten-mu}
  \mu = 1 + \frac{2\hat{m}}{\sigma^2}.
\end{equation}
The matrix Kesten variable is an inverse-Wishart matrix, whose eigenvalue spectrum maps to Marchenko--Pastur under $\lambda \to 1/\lambda$.
This establishes a precise duality: the additive regime is Marchenko--Pastur; the multiplicative regime is inverse-Wishart; the two are related by spectral inversion.

\begin{remark}[Interpolating Family]
The free log-normal provides a one-parameter family interpolating continuously between additive and multiplicative regimes.
At small $a$ (short horizon), the density resembles the Marchenko--Pastur semicircle.
At large $a$ (long horizon), the density develops a power-law tail controlled by $\mu = 1 + 2\hat{m}/\sigma^2$.
The parameter $a$ plays the role of effective horizon, growing with the number of compounding periods.
\end{remark}

\begin{remark}[Generalised Dyson Brownian Motion]
\citet{BousseyroBouchaud2025} extend the Dyson Brownian motion
framework by interpreting the additivity of R-transforms in free
probability as a dynamical evolution of eigenvalues interacting through
two-body and higher-body forces.  Their generalised Dyson equation
accommodates motion matrices $\mathbf{B}$ with arbitrary free cumulants
$\kappa_n^{\mathbf{B}}$---not just the Gaussian case ($n = 2$) of
standard DBM.  When $\kappa_n^{\mathbf{B}} \neq 0$ for $n \geq 3$, the
eigenvalue dynamics acquire $n$-body repulsion terms that modify both
the bulk density and the conditions under which outlier eigenvalues
detach.  Since asset returns are heavy-tailed, this extension is
relevant to the multiplicative DBM of equation~\eqref{eq:mult-dbm}:
replacing Gaussian noise $\xi_i$ with $\alpha$-stable driving produces
a different family of stationary spectral densities whose tail exponent
depends on both the compounding horizon and the noise stability
index~$\alpha$.  The free-convolution machinery also provides a natural
algebraic setting for computing the spectral density of products of free
random matrices---the operation that arises when portfolio allocations
are composed across time periods or aggregated across investors.
\end{remark}

For the spectral portfolio framework:
\begin{enumerate}
  \item The \emph{within-layer} spectral structure of SGD weight matrices reflects the additive regime.
    The network is trained on short-horizon data increments; the relevant universality class is Marchenko--Pastur; the relevant dynamics are additive DBM.
  \item The \emph{cross-sectional} spectral structure of the population allocation matrix reflects the multiplicative regime.
    Agents compound returns over long horizons; the relevant universality class is the free log-normal/inverse-Wishart; the dynamics are multiplicative DBM.
  \item The aggregation from matrix to scalar (Section \ref{sec:aggregation}) is the radial projection from matrix Fokker--Planck to scalar Fokker--Planck, mapping the spectral exponent $\mu$ of inverse-Wishart to the Pareto exponent $\alpha = 1 + v/D$ of the wealth distribution.
  \item The scalar Kesten process of \citet{Froeseth2026S} and the matrix Kesten problem are related by this same projection.
    The Pareto tail of the wealth distribution is the radial shadow of the inverse-Wishart tail of the allocation matrix.
\end{enumerate}

Regime shifts manifest spectrally as changes in the effective parameter $a$ of the free log-normal, from small $a$ (Marchenko--Pastur bulk with isolated spikes) to large $a$ (inverse-Wishart density with power-law tails).
Market crashes or policy reforms can trigger cross-sectional regime shifts: assets in the additive regime may suddenly transition to multiplicative dynamics (fire sales, forced liquidation).

\section{The Ergodicity Gap as Spectral Observable}\label{sec:ergodicity}

A central insight from \citet{Peters2019} and \citet{Froeseth2026S} is the distinction between \emph{ensemble average} and \emph{time average}.
The wealth distribution $\pi(x)$ represents the ensemble (cross-section at a fixed time).
An individual trajectory $x(t)$ represents time evolution of a single agent.

In ergodic systems, the two averages coincide: an agent's long-run average outcome equals the population average.
In non-ergodic systems, they differ: the population distribution can have fat tails while individual trajectories concentrate.

\begin{definition}[Ergodicity Gap]\label{def:ergo-gap}
For a wealth process $x(t)$ with drift $v$ and diffusion $D$, define the \emph{ergodicity gap} as the difference between the ensemble and time-average growth rates:
\begin{equation}\label{eq:ergo-gap}
  \Delta g \;=\; g_{\text{ens}} - g_{\text{time}}
  \;=\; \bigl(v + \tfrac{1}{2}D\bigr) - \bigl(v - \tfrac{1}{2}D\bigr)
  \;=\; D.
\end{equation}
Here $g_{\text{ens}} = \frac{d}{dt}\log\mathbb{E}[x]$ is the growth rate of the ensemble mean and $g_{\text{time}} = \mathbb{E}[\frac{d}{dt}\log x]$ is the expected time-average (geometric) growth rate.
\end{definition}

In the scalar model of \citet{Froeseth2026S}, the ergodicity gap is simply the diffusion coefficient $D$.
In the spectral framework, the Itô projection of Section~\ref{sec:aggregation} yields an effective scalar diffusion
\begin{equation}\label{eq:Deff-spectral}
  D_{\text{eff}} = 2\eta D,
\end{equation}
where $\eta$ is the learning rate and $D$ the data covariance scale.
The radial projection collapses all $mn$ noise directions onto one: each matrix entry contributes independently to $\|W\|_F^2$, but the quadratic variation of $x = \|W\|_F$ is $(1/x^2)\sum W_{ai}^2 \cdot 2\eta D = 2\eta D$.
The ergodicity gap thus has a spectral interpretation:
\begin{equation}\label{eq:ergo-spectral}
  \Delta g = D_{\text{eff}} = 2\eta D.
\end{equation}
The gap depends on the noise level but not on matrix dimensions.
However, the \emph{effect} of the gap on the wealth distribution depends on the spectral structure: when the spectral tail is fat (small tail exponent), fewer factors carry most of the variance, individual trajectories are more volatile relative to their mean, and the distributional consequences of non-ergodicity are more severe.
When the spectrum is flat (many comparable factors), the volatility is spread across dimensions and individual trajectories track the ensemble more closely.

An isotropic perturbation (Theorem~\ref{thm:spectral-invariance}) preserves the spectral shape, hence preserves $D_{\text{eff}}$ and leaves the ergodicity gap unchanged.
An anisotropic perturbation that compresses the spectral tail (e.g.\ differential taxation favouring diversified portfolios) reduces the effective number of dominant factors, lowers $D_{\text{eff}}$, and narrows the gap.

\section{Aggregation: From Matrix to Scalar via Itô's Lemma}\label{sec:aggregation}

Define total portfolio value as $x = \|W\|_F = \bigl(\sum_{a,i} W_{ai}^2\bigr)^{1/2}$.
By Itô's lemma applied to $f(W) = \|W\|_F$:
\begin{equation}\label{eq:ito-frobenius}
  dx = \frac{1}{x}\sum_{a,i} W_{ai}\,dW_{ai}
       + \frac{1}{2x}\Bigl(
         mn \cdot 2\eta D
         - \frac{1}{x^2}\sum_{a,i} W_{ai}^2 \cdot 2\eta D
       \Bigr)dt.
\end{equation}

The first term gives the drift contribution from the gradient signal projected onto the radial direction.
The second term simplifies using $\sum W_{ai}^2 = x^2$:
\begin{equation}\label{eq:radial-sde}
  dx = \frac{1}{x}\operatorname{tr}(W^\top dW)
       + \frac{\eta D(mn - 1)}{x}\,dt.
\end{equation}

In SVD coordinates, $\operatorname{tr}(W^\top dW) = \sum_k \sigma_k\,d\sigma_k$ plus rotational terms.
If the rotational terms average out (which they do under isotropic noise), the radial process becomes
\begin{equation}\label{eq:radial-gbm}
  \frac{dx}{x} = \mu_{\text{rad}}\,dt + \sigma_{\text{rad}}\,dB,
\end{equation}
a geometric Brownian motion.
This is exactly the scalar SDE of \citet{Froeseth2026S}.
\textbf{The wealth Fokker--Planck is the radial projection of the matrix Fokker--Planck.}

The radial drift and diffusion are computed by substituting the multiplicative Dyson BM~\eqref{eq:mult-dbm} into~\eqref{eq:radial-sde}.
Let $p = \min(m, n)$ be the number of singular values and write $x^2 = \sum_{k=1}^p \sigma_k^2$.
The signal term contributes a radial drift from the restoring force in~\eqref{eq:stat-sv}, and the Itô correction from~\eqref{eq:ito-frobenius} contributes $\eta D(mn - 1)/x$.
The noise has quadratic variation $(1/x^2)\sum_{a,i}\|W_{ai}\|^2 \cdot 2\eta D = 2\eta D$, giving an effective scalar diffusion $D_{\text{eff}} = 2\eta D$ independent of matrix dimensions (the radial projection collapses all $mn$ noise directions onto one).
The Pareto exponent is
\begin{equation}\label{eq:pareto-spectral}
  \alpha = 1 + \frac{v_{\text{radial}}}{D_{\text{eff}}},
\end{equation}
where $v_{\text{radial}}$ depends on $m$, $n$, and $\beta_1/\eta D$.

The function $f$ can be evaluated explicitly in two regimes.

\textbf{Single-factor limit} ($p = 1$, i.e.\ $n = 1$).
The allocation matrix reduces to a vector, eigenvalue repulsion vanishes, and the spectral SDE collapses to a scalar Ornstein--Uhlenbeck process for $\sigma^2$.
The radial SDE~\eqref{eq:radial-gbm} becomes
\[
  dx = \Bigl(-\frac{\beta_1}{2}x + \frac{\eta D(m-1)}{x}\Bigr)dt + \sqrt{2\eta D}\,dB.
\]
At stationarity, the process for $y = x^2$ has drift $v_y = -\beta_1 y + 2\eta D \cdot m$ and diffusion $4\eta D \cdot y$.
The stationary density is $P(y) \propto y^{-1-\mu}e^{-\text{const}/y}$ with
\begin{equation}\label{eq:pareto-single}
  \mu\big|_{p=1} = 1 + \frac{(m - 1)\beta_1}{2\eta D},
\end{equation}
which is exactly the scalar Kesten exponent~\eqref{eq:matrix-kesten-mu} with $\hat{m} = (m-1)\beta_1/2$ and $\sigma^2 = 2\eta D$.
The corresponding wealth Pareto exponent is $\alpha = \mu/2$ (since wealth scales as $x$, not $x^2$).

\textbf{Large-$N$ regime} ($m, n \gg 1$, $m/n = \gamma$ fixed).
The eigenvalue repulsion is of order $p^2$ but the repulsion terms cancel in the radial projection (they redistribute variance among singular values without changing $\sum \sigma_k^2$).
The radial drift at stationarity is $v_{\text{radial}} = -\beta_1 x/2 + \eta D(mn - 1)/x$.
Setting $v_{\text{radial}} = 0$ at the stationary mean $\bar{x}^2 = 2\eta D(mn - 1)/\beta_1$ and linearising, the Pareto exponent becomes
\begin{equation}\label{eq:pareto-largeN}
  \alpha \;\approx\; 1 + \frac{(mn - 1)\beta_1}{2mn\,\eta D}
  \;\approx\; 1 + \frac{\beta_1}{2\eta D},
  \qquad m, n \gg 1.
\end{equation}
In the large-$N$ limit the Pareto exponent depends on the signal-to-noise ratio $\beta_1/(2\eta D)$ alone, independent of matrix dimensions.
This is the spectral counterpart of the Kesten result: strong signal (large $\beta_1$) produces a thin tail (large $\alpha$), while strong noise (large $\eta D$) produces a fat tail (small $\alpha$).

For general $(m, n)$, the function $f$ interpolates between the single-factor formula~\eqref{eq:pareto-single} and the large-$N$ limit~\eqref{eq:pareto-largeN}, and can be computed by numerical integration of the radial projection against the spectral density~\eqref{eq:stat-sv}.

\section{The Fokker--Planck Connection}\label{sec:fp-connection}

Both the wealth distribution and the spectral density satisfy Fokker--Planck equations, revealing a deep structural parallel.

\textbf{Wealth distribution \citep{Froeseth2026S}:}
\begin{equation}\label{eq:fp-wealth}
  \frac{\partial \pi}{\partial t}
    = -\frac{\partial}{\partial x}\bigl[v(x)\,\pi\bigr]
    + D\,\frac{\partial^2 \pi}{\partial x^2}.
\end{equation}
The stationary solution has Pareto tail with exponent $\alpha = 1 + v/D$.

\textbf{Spectral density:}
\begin{equation}\label{eq:fp-spectral}
  \frac{\partial \rho}{\partial \lambda}
    = \frac{\partial}{\partial \lambda}\Bigl[
      \bigl(c(\lambda^* - \lambda)
      - \mathrm{PV}\!\int \frac{\rho(\lambda')}{\lambda - \lambda'}\,d\lambda'\bigr)\rho
    \Bigr]
    + D_{\text{eff}}\,\frac{\partial^2 \rho}{\partial \lambda^2}.
\end{equation}
The stationary solution has gamma-type density with power-law tail exponent $(m - n + 1)/2$.

The spectral equation has an extra nonlocal interaction term (the principal-value integral---the mean-field version of eigenvalue repulsion).
Structurally, both are drift--diffusion PDEs with power-law stationary solutions.

If we consider a single-factor portfolio ($N = 1$, so the allocation matrix reduces to a vector), the repulsion term vanishes and the spectral FP reduces to the standard wealth FP.
In this limit, the spectral exponent should reduce to the Pareto exponent.
The scalar wealth model is the $N = 1$ projection of the matrix model.
Figure~\ref{fig:spectral-flow} summarises the flow from the three forces through the spectral density to the wealth distribution.

\begin{figure}[H]
\centering
\begin{tikzpicture}[
    >=Stealth,
    lbl/.style={font=\scriptsize, text=black!60},
    arrow/.style={->, thick, shorten >=2pt, shorten <=2pt},
  ]

  \node[font=\small\bfseries] at (0, 5.5) {Three forces on singular values $\sigma_k$};

  \node[draw, rounded corners=4pt, thick, fill=blue!8, minimum width=3.2cm,
        minimum height=0.9cm, font=\scriptsize, align=center] (f1) at (-4.0, 4.5)
    {\textbf{Gradient signal}\\$-\eta\,u_k^\top(\nabla_W\mathcal{L})\,v_k$};
  \node[draw, rounded corners=4pt, thick, fill=green!8, minimum width=3.2cm,
        minimum height=0.9cm, font=\scriptsize, align=center] (f2) at (0, 4.5)
    {\textbf{Survival constraint}\\$\eta D\,(m{-}n{+}1)/(2\sigma_k)$};
  \node[draw, rounded corners=4pt, thick, fill=red!8, minimum width=3.2cm,
        minimum height=0.9cm, font=\scriptsize, align=center] (f3) at (4.0, 4.5)
    {\textbf{Eigenvalue repulsion}\\$+\eta D\sum_{j\neq k}\sigma_k/(\sigma_k^2{-}\sigma_j^2)$};

  \node[lbl] at (-4.0, 3.75) {smart money};
  \node[lbl] at (0, 3.75) {minimum position};
  \node[lbl] at (4.0, 3.75) {diversification};

  \draw[arrow, blue!50] (f1.south) -- (-4.0, 3.55) -- (-1.5, 2.7);
  \draw[arrow, green!50!black] (f2.south) -- (0, 2.7);
  \draw[arrow, red!50] (f3.south) -- (4.0, 3.55) -- (1.5, 2.7);

  \node[font=\small\bfseries] at (0, 2.4) {Stationary spectral density: $p(\sigma) \propto \sigma^{(m-n+1)/2}\,e^{-\beta_1\sigma^2/4\eta D}$};

  \begin{scope}[xshift=0cm, yshift=0.3cm]
    \draw[->, thick, black!50] (-4.2, -0.5) -- (4.5, -0.5) node[right, lbl] {$\sigma$};
    \draw[->, thick, black!50] (-4.2, -0.5) -- (-4.2, 1.5) node[left, lbl] {$p(\sigma)$};

    \draw[very thick, blue!70!black, smooth]
      plot[domain=-3.8:3.8, samples=80]
      ({\x}, {2.0*exp(-0.5*(\x+1.0)*(\x+1.0)) + 0.3*exp(-0.08*(\x-0.5)*(\x-0.5))*max(0,\x+1.0)/4.8 - 0.5});

    \draw[decorate, decoration={brace, amplitude=4pt, mirror}, thick, black!50]
      (-3.5, -0.7) -- (-0.2, -0.7)
      node[midway, below=5pt, font=\scriptsize] {\textbf{bulk} (core)};

    \draw[decorate, decoration={brace, amplitude=4pt, mirror}, thick, black!50]
      (0.2, -0.7) -- (3.8, -0.7)
      node[midway, below=5pt, font=\scriptsize] {\textbf{tail} $\sim\!\sigma^{(m-n+1)/2}$ (satellites)};
  \end{scope}

  \draw[arrow, very thick, black!60] (0, -1.5) -- (0, -2.5)
    node[midway, right=4pt, font=\scriptsize, text=black!60, align=left]
    {Itô projection\\$x = \|W\|_F$};

  \node[font=\small\bfseries] at (0, -3.0) {Wealth distribution: $\pi(x) \propto x^{-1-\alpha}$, \quad $\alpha = 1 + v_{\mathrm{radial}}/D_{\mathrm{eff}}$};

  \begin{scope}[xshift=0cm, yshift=-4.5cm]
    \draw[->, thick, black!50] (-4.2, -0.6) -- (4.5, -0.6) node[right, lbl] {$x$ (wealth)};
    \draw[->, thick, black!50] (-4.2, -0.6) -- (-4.2, 1.0) node[left, lbl] {$\pi(x)$};

    \draw[very thick, orange!70!red, smooth]
      plot[domain=-3.5:3.8, samples=60]
      ({\x}, {min(0.85, 1.8/max(0.5, (\x+4.0))*exp(-0.02*(\x+4.0)*(\x+4.0))) - 0.6});

    \draw[<-, thick, black!50] (2.5, -0.35) -- (3.5, 0.3)
      node[right, font=\scriptsize, align=left] {Pareto tail\\$\alpha = 1 + \frac{\beta_1}{2\eta D}$};
  \end{scope}

  \node[draw, rounded corners=3pt, thick, fill=yellow!10, font=\scriptsize,
        text width=2.8cm, align=center] at (6.8, -2.0)
    {spectral exponent\\$(m - n + 1)/2$\\[3pt]$\downarrow$\\[3pt]Pareto exponent\\$\alpha$};

\end{tikzpicture}
\caption{From forces to wealth distributions.
The three forces in the singular-value SDE~\eqref{eq:sv-sde} --- gradient signal, survival constraint, and eigenvalue repulsion --- determine the stationary spectral density (middle).
The gamma-type density has a concentrated bulk (core portfolio) and a power-law tail (satellite positions).
The radial Itô projection (Section~\ref{sec:aggregation}) maps the matrix-valued spectral density to a scalar wealth process $x = \|W\|_F$, whose stationary distribution exhibits a Pareto tail with exponent $\alpha$ determined by the signal-to-noise ratio $\beta_1/\eta D$.}
\label{fig:spectral-flow}
\end{figure}
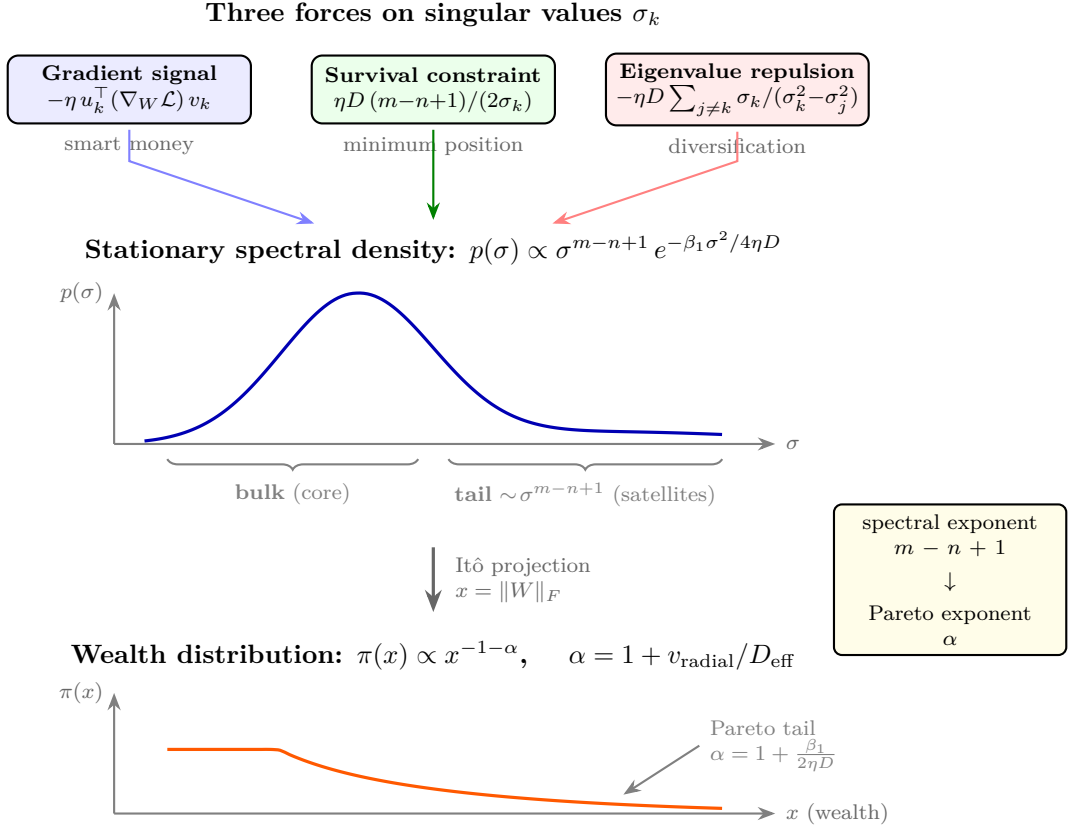

\section{The Bouchaud--Mézard Unification}\label{sec:bouchaud}

\citet{BouchaudMezard2000} model wealth dynamics across a population of agents:
\begin{equation}
  dw_i = \eta_i w_i\,dt + \sum_{j} J_{ij}(w_j - w_i)\,dt,
\end{equation}
where $\eta_i$ is multiplicative noise (investment returns) and $J_{ij}$ describes wealth exchange (trade, lending, redistribution) between agents.

Stack $N$ agents' wealths into a vector $w = (w_1, \ldots, w_N)$.
The exchange term is $Jw - \operatorname{diag}(J\mathbf{1})w$, a matrix--vector product minus diagonal correction.
The full dynamics:
\begin{equation}
  dw = \bigl[
    \operatorname{diag}(\eta) + J - \operatorname{diag}(J\mathbf{1})
  \bigr]\,w\,dt.
\end{equation}

The matrix $M = \operatorname{diag}(\eta) + J - \operatorname{diag}(J\mathbf{1})$ has eigenvalues whose distribution determines the steady-state wealth distribution.
\citet{BouchaudMezard2000} show that the Pareto exponent is determined by the balance between the noise spectrum (eigenvalues of $\operatorname{diag}(\eta)$) and the exchange coupling $J$.

The exchange matrix $J$ introduces \textbf{inter-agent coupling}---analogous to the eigenvalue repulsion in \citet{OlsenEtAl2025}.

\begin{itemize}
  \item Without coupling ($J = 0$ or no repulsion): wealth and singular values evolve independently, leading to condensation.
  \item With coupling: repulsive interaction prevents condensation and produces power-law tails.
\end{itemize}

This establishes a \textbf{unified framework} in which three descriptions of wealth-related dynamics share common spectral structure:

\begin{enumerate}
  \item Bouchaud--Mézard describes the \emph{cross-sectional} dynamics (wealth distribution across agents).
  \item \citet{OlsenEtAl2025} describes the \emph{within-portfolio} dynamics (allocations across assets).
  \item The scalar Fokker--Planck framework of \citet{Froeseth2026S} provides the \emph{radial projection} linking spectral and Pareto exponents.
\end{enumerate}

All three are instances of interacting particle systems with multiplicative noise, where the interaction term prevents condensation and produces power-law distributions.
Any uniform perturbation (e.g.\ a tax, a fee, a regulatory cost) modifies the multiplicative noise term in all three descriptions simultaneously.

The binary condensation/no-condensation dichotomy above admits a richer
phase structure when agents have heterogeneous growth rates.
\citet{BernardBouchaudLeDoussal2026} solve the mean-field version of the
Bouchaud--Mézard model with quenched disorder in the drift $\eta_i$ and
identify three phases as a function of the redistribution rate~$\varphi$:
(i)~a \emph{localised} phase ($\varphi < \varphi_c$), in which a single
agent captures a macroscopic fraction of total wealth---the analogue of
Bose--Einstein condensation; (ii)~a \emph{partially localised} phase, in
which the wealth distribution has a power-law tail with exponent
$\mu = 1 - \Sigma_0^2/\sigma^4$, where $\Sigma_0^2$ is the variance of
the quenched growth rates and $\sigma^2$ the diffusion intensity; and
(iii)~a \emph{delocalised} phase ($\varphi > \varphi_c$), in which wealth
is spread across agents.  The critical redistribution rate $\varphi_c$
marks a genuine phase transition, derived using an analogy with Derrida's
Random Energy Model.  In the spectral language, the localised phase
corresponds to a participation ratio $p_1 \to 1 - \varphi/\varphi_c$, so
that a single eigenmode dominates the wealth vector---a spectral
fingerprint of condensation that is absent in the delocalised phase.

\section{Isotropic Perturbations and Spectral Invariance}\label{sec:invariance}

We now establish the central invariance result.
Consider any perturbation to the portfolio objective that affects all assets uniformly---what we call an \emph{isotropic perturbation}.
Examples include a proportional wealth tax with uniform assessment, a flat management fee, a uniform transaction cost, or any regulatory levy that does not distinguish between asset classes.

\subsection{The Isotropic Condition}\label{subsec:isotropic-condition}

An isotropic perturbation modifies the loss function as
\begin{equation}\label{eq:loss-rescale}
  \mathcal{L}_{\text{post}}(W)
    = k \cdot \mathcal{L}_{\text{pre}}(W) + \text{const},
\end{equation}
where $k > 0$ is a scalar that depends on the perturbation parameters but not on the direction in weight space.
This is a uniform rescaling of the loss across all directions.

\subsection{Spectral Consequences}\label{subsec:spectral-consequences}

This uniform rescaling produces three immediate consequences:

\begin{enumerate}
  \item \textbf{Gradient Scaling.}
    The gradient rescales uniformly: $\nabla_W \mathcal{L}_{\text{post}} = k \nabla_W \mathcal{L}_{\text{pre}}$.
    All signal terms in the singular-value SDE scale by the same factor $k$.

  \item \textbf{Repulsion Structure Preserved.}
    The eigenvalue repulsion term depends only on $D$ and the geometry of the singular value space, not on the loss function.
    It is unaffected by the perturbation.

  \item \textbf{Spectral Density Transformation.}
    The stationary singular-value density becomes
    \begin{equation}
      p_\sigma^{\text{post}}(\sigma) = p_\sigma^{\text{pre}}\bigl(\sigma/\sqrt{k}\bigr) / \sqrt{k},
    \end{equation}
    a scale-shifted version of the pre-perturbation distribution.
    The shape (the exponent $(m - n + 1)/2$ in the power-law tail) is preserved.
\end{enumerate}

\begin{theorem}[Spectral Invariance]\label{thm:spectral-invariance}
An isotropic perturbation to the portfolio objective preserves the singular-value distribution of the allocation matrix up to a scale-and-shift transformation.
The tail exponent, the eigenportfolio directions, and the effective spectral rank are all invariant.
\end{theorem}

\begin{proof}[Proof sketch]
Let $W \in \mathbb{R}^{T \times N}$ be the allocation matrix with singular-value decomposition $W = U \Sigma V^\top$, where $\Sigma = \operatorname{diag}(\sigma_1, \ldots, \sigma_{\min(T,N)})$.
The singular-value dynamics under SGD follow the multiplicative Dyson Brownian motion~\eqref{eq:mult-dbm}:
\[
  \frac{\dd\sigma_i}{\dd t}
    = \underbrace{g(\sigma_i)}_{\text{signal}}
    + \underbrace{\frac{b}{N}\sum_{j \neq i}\frac{\sigma_i\sigma_j}{\sigma_i - \sigma_j}}_{\text{repulsion}}
    + \underbrace{\sqrt{\frac{b}{N}}\,\sigma_i\,\xi_i}_{\text{noise}},
\]
where $g(\sigma_i)$ encodes the gradient signal from the loss function.

\emph{Step 1: Gradient rescaling.}
Under the isotropic perturbation~\eqref{eq:loss-rescale}, $\nabla_W \mathcal{L}_{\text{post}} = k \nabla_W \mathcal{L}_{\text{pre}}$, so the signal term transforms as $g \to k\,g$.
The repulsion term depends only on the spectral geometry of $W^\top W$ and is independent of the loss; the noise term is determined by the learning rate and data covariance.
Both are unaffected.

\emph{Step 2: Rescaling.}
Define $\tilde{\sigma}_i = \sigma_i / \sqrt{k}$.  By It\^{o}'s lemma,
\[
  \frac{\dd\tilde{\sigma}_i}{\dd t}
    = g(\tilde{\sigma}_i)
    + \frac{b}{N}\sum_{j \neq i}\frac{\tilde{\sigma}_i\tilde{\sigma}_j}{\tilde{\sigma}_i - \tilde{\sigma}_j}
    + \sqrt{\frac{b}{N}}\,\tilde{\sigma}_i\,\xi_i,
\]
which is the original SDE.  Hence $\tilde{\sigma}_i$ has the same stationary distribution as $\sigma_i^{\text{pre}}$.

\emph{Step 3: Spectral invariants.}
The stationary density transforms as $p_\sigma^{\text{post}}(\sigma) = p_\sigma^{\text{pre}}(\sigma/\sqrt{k})/\sqrt{k}$, confirming~(5).
The tail exponent $\mu = 1 + 2\hat{m}/\sigma^2$ depends on $\hat{m}$ and $\sigma^2$ only through their ratio, which is scale-invariant.
The eigenportfolio directions $v_k$ (columns of $V$) are determined by the angular part of the SVD, which is unchanged by a scalar rescaling of the loss.
The effective spectral rank $r_{\text{eff}} = (\sum \sigma_i)^2 / \sum \sigma_i^2$ is a ratio of homogeneous functions and is therefore scale-invariant.
\end{proof}

Since the portfolio weights are defined as $w_i = W_{ti} / \sum_j W_{tj}$ (allocation proportion in state $t$), scaling all of $W$ by a constant $k$ leaves the proportions unchanged:
\begin{equation}
  w^{\text{post}} = w^{\text{pre}}.
\end{equation}
The portfolio composition is invariant under any isotropic perturbation.

\begin{corollary}[Tax Neutrality under Homogeneous Returns]\label{cor:tax-neutral}
The combined-tax stack of \citet{Froeseth2026F}---corporate tax~$\tau_c$, capital income tax~$\tau_k$, dividend tax~$\tau_d$ with shielding rate~$\rho_s$, and proportional wealth tax~$\tw$ with assessment fraction~$\alpha$---satisfying neutrality conditions C1 ($\tau_k = \tau_c$), C2 ($\rho_s = r_f$), and C3 ($\alpha_i = \alpha$ for all~$i$), preserves the cross-sectional spectral structure of the allocation matrix under the homogeneous-return assumption that all rows of $W$ share a common drift vector~$\boldsymbol{\mu}$ and covariance~$\V$.
The Pareto exponent of the cross-sectional wealth distribution, the eigenportfolio directions (right singular vectors of $W$), and the effective spectral rank are invariant under the tax stack.
The deterministic equilibrium $\bar W$ shifts by a rank-one matrix uniform across rows; the stochastic fluctuations $\tilde W = W - \bar W$ have an unchanged stationary distribution.
This is the spectral expression of the drift-shift symmetry of \citet{Froeseth2026S} generalised to the combined-tax stack.
\end{corollary}

\begin{proof}
The proof uses a different mechanism than Theorem~\ref{thm:spectral-invariance}: the tax stack is a return-rescale (with curvature $\V$ unchanged), not a loss-rescale, so the form~\eqref{eq:loss-rescale} does not apply directly.
Instead, we exploit $\V$-preservation under $\mathcal{T}^{\mathrm{gen}}$ together with the row-uniform action implied by the homogeneous-return assumption.

Let $\mathcal{L}(W) = -\sum_{t,i} W_{ti}(\mu_i - r_f) + \tfrac{\gamma}{2}\sum_t W_{t\cdot}^\top \V\,W_{t\cdot}$ denote the mean--variance loss in excess-return form, with $\bar W$ the deterministic first-order condition of $\nabla_W\mathcal{L} = 0$ and $\tilde W = W - \bar W$ the stochastic fluctuation around it.
Under the SGD dynamics~\eqref{eq:olsen-sde} with isotropic noise of intensity~$D$, $\tilde W$ obeys
\[
  d\tilde W = -\eta\gamma\,\tilde W\,\V\,dt + \sqrt{2\eta D}\,d\mathcal{W}.
\]

By Theorem~1 of \citet{Froeseth2026F}, conditions C1--C3 (with $\alpha_0 = 0$ for the risk-free asset, i.e.\ the wealth tax exempts the risk-free leg) imply that the after-tax excess return on every asset~$i$ takes the form
\[
  R_i^{\text{ex,post}} = \lambda\,(\mu_i - r_f) + c, \qquad
  \lambda := (1-\tau_c)(1-\tau_d), \quad
  c := -(1-\tau_d)\tau_c r_f - \tw\alpha,
\]
where the additive constant~$c$ is the same for all assets~$i$ (by C3) and, under the homogeneous-return assumption, the same for all rows~$t$ of the allocation matrix.
Diffusion is preserved: $\V^{\text{post}} = \V^{\text{pre}}$, in line with the drift-shift symmetry of \citet{Froeseth2026S}.

The fluctuation SDE depends on $\V$ only, not on $\boldsymbol{\mu}$ or~$c$: both the linear-in-$\tilde W$ drift coefficient $-\eta\gamma\V$ and the noise $\sqrt{2\eta D}\,d\mathcal{W}$ are unchanged under the tax stack.
The SDE for $\tilde W$ is therefore identical pre and post-tax; the stationary distribution of fluctuations is invariant.

The deterministic FOC, computed from $\nabla_{W_{t\cdot}}\mathcal{L} = 0$ on the mean--variance loss with effective excess return $\boldsymbol{R}^{\text{ex,post}} = \lambda(\boldsymbol{\mu} - r_f\one) + c\,\one$, shifts uniformly across rows:
\[
  \bar W^{\text{post}} - \bar W^{\text{pre}}
    = \one_T \cdot \boldsymbol{\delta}^\top, \qquad
  \boldsymbol{\delta} = \frac{1}{\gamma}\,\V^{-1}\bigl((\lambda-1)(\boldsymbol{\mu} - r_f\one) + c\,\one_N\bigr),
\]
a rank-one matrix common to every row, corresponding to a uniform translation of log-wealth in the population.

Therefore: (i)~the stationary cross-sectional distribution of $\tilde W$ is invariant, so the Pareto tail exponent, eigenportfolio directions, and effective spectral rank of the stationary spectrum are preserved; (ii)~the cross-sectional mean shifts uniformly, in line with drift-shift symmetry.
The homogeneous-return assumption is essential: under heterogeneous returns (different $\boldsymbol{\mu}_t$ per row, as in \citet{Froeseth2026H}), the FOC shift is no longer row-uniform and cross-sectional shape need not be preserved.
\end{proof}

\section{Anisotropic Perturbations and Spectral Distortion}\label{sec:anisotropic}

When a perturbation affects different assets differently, the loss landscape is modified \emph{anisotropically}:
\begin{equation}\label{eq:aniso-perturbation}
  \frac{\partial \mathcal{L}}{\partial W_{ti}} \to \frac{\partial \mathcal{L}}{\partial W_{ti}} + \delta_i\,W_{ti},
\end{equation}
where $\delta_i$ is the asset-specific perturbation strength.
This is the regime \citet{OlsenEtAl2025} identify (Proposition 6.17) as breaking their isotropic analysis and producing distorted spectra.

The consequences are:
\begin{itemize}
  \item Eigenportfolio directions $v_k$ rotate (portfolio tilt toward less-perturbed assets).
  \item The repulsion structure changes (some factors compressed, others stretched).
  \item The tail exponent becomes direction-dependent.
\end{itemize}

The portfolio weight distortion, to first order in the perturbation, is
\begin{equation}\label{eq:weight-distortion}
  \Delta w^* = -\frac{1}{\gamma} V^{-1}(\delta - \bar{\delta} \cdot \mathbf{1}),
\end{equation}
where $\gamma$ is risk aversion, $V$ is the return covariance matrix, $\delta = (\delta_1, \ldots, \delta_N)$ is the vector of perturbation strengths, and $\bar{\delta}$ is their mean.

\begin{example}[Differential Taxation]
In the Norwegian wealth tax system, $\delta_i = \tau_w(1 - \alpha_i)$ where $\alpha_i$ is the assessment fraction for asset class $i$.
With $\alpha_{\text{housing}} = 0.25$, $\alpha_{\text{shares}} = 0.80$, $\alpha_{\text{deposits}} = 1.00$, the cross-asset dispersion of perturbation strengths $\operatorname{Var}(1 - \alpha_i) \approx 0.10$ drives portfolio tilt toward housing.
Post-2017 discount changes (which changed $\alpha_{\text{shares}}$ from 1.00 to 0.80 over several steps) should produce time-varying spectral distortions testable in SSB microdata.
\end{example}

\begin{example}[Sector-Specific Regulation]
A financial regulation imposing different capital requirements on different sectors (e.g.\ Basel risk weights) enters as $\delta_i = c_i$ where $c_i$ is the capital charge per unit of exposure to sector $i$.
The spectral distortion is proportional to $\operatorname{Var}(c_1, \ldots, c_N)$.
\end{example}

\begin{example}[Differential Transaction Costs]
Market microstructure costs that vary across asset classes (e.g.\ bid--ask spreads, brokerage fees) enter as $\delta_i = \kappa_i$ where $\kappa_i$ is the effective cost of trading asset $i$.
Illiquid assets (high $\kappa_i$) are underweighted relative to the frictionless optimum.
\end{example}

\section{The Spectral Invariance Conjecture}\label{sec:conjecture}

\begin{conjecture}[Spectral Invariance]\label{conj:sic}
Consider $N$ assets with returns observed over $T$ periods.
Let the allocation matrix $W \in \mathbb{R}^{T \times N}$ evolve under stochastic gradient descent on a portfolio objective $\mathcal{L}(W)$ with isotropic noise of intensity $D$.

\begin{enumerate}
  \item[\textup{(a)}]
    The stationary singular-value distribution has tail exponent $(T - N + 1)/2$, and the aggregate wealth distribution formed by $x = \|W\|_F$ has Pareto tail exponent $\alpha = 1 + v/D$ where $v$ is the mean drift of $\|W\|_F$.

  \item[\textup{(b)}]
    Any isotropic perturbation to $\mathcal{L}$ preserves the spectral shape: the singular-value distribution shifts in scale but does not change its tail exponent or eigenportfolio directions.

  \item[\textup{(c)}]
    Any anisotropic perturbation to $\mathcal{L}$ with asset-specific weights $\delta_i$ rotates the eigenportfolio directions and changes the tail exponent in a direction-dependent way.
    The magnitude of distortion is proportional to the cross-asset variance:
    \begin{equation}
      \textup{distortion} \propto \operatorname{Var}(\delta_1, \ldots, \delta_N).
    \end{equation}
\end{enumerate}
\end{conjecture}

\textbf{Status.}
Parts (a) and (b) follow from combining \citet{OlsenEtAl2025}'s theorems with the isotropic rescaling argument of Theorem~\ref{thm:spectral-invariance}.
Part (c) requires extending \citet{OlsenEtAl2025}'s analysis to the anisotropic case; their Proposition 6.17 is a starting point but does not give the full spectral distortion.

\begin{corollary}[Tax Neutrality]
Setting $\delta_i = \tau_w(1 - \alpha_i)$, part~(b) recovers the neutrality conditions C1--C3 of \citet{Froeseth2026N} and part~(c) quantifies the cost of violating condition C3.
\end{corollary}

\textbf{Interpretation.}
The Conjecture provides a unified criterion for evaluating any environmental perturbation---tax, regulation, transaction cost, or fee structure---in terms of its spectral signature.
Part~(a) connects spectral and Pareto exponents; part~(b) characterises benign (isotropic) perturbations; part~(c) quantifies distortion from differential treatment.
The free-convolution framework of \citet{BousseyroBouchaud2025}---which
interprets the additivity of R-transforms as a generalised Dyson
evolution---provides a natural algebraic setting for attacking parts~(a)
and~(b), since the spectral density of the allocation matrix at finite
horizon is a free convolution of the initial condition with the noise
contribution.

\section{The Loss Function and Utility Correspondence}\label{sec:loss}

The loss function $\mathcal{L}(W)$ encodes the portfolio objective.
Different loss functions (MLE, score matching, denoising score matching) encode different utility assumptions.

\subsection{Maximum Likelihood Estimation}

For maximum likelihood estimation of a drift function, the loss is
\begin{equation}
  \mathcal{L}_{\text{MLE}}(W) = -\frac{1}{T}\sum_{t=1}^{T} \log p(x_{t+1} \mid x_t; W).
\end{equation}
Under the SDE model $dx = \hat{v}(x; W)\,dt + \sqrt{2D}\,dB$, the transition density is approximately Gaussian:
\begin{equation}\label{eq:mle-loss}
  \mathcal{L}_{\text{MLE}}(W) = \frac{1}{4DT\Delta t}\sum_t \bigl(x_{t+1} - x_t - \hat{v}(x_t; W)\,\Delta t\bigr)^2 + \text{const}.
\end{equation}
This is a \textbf{quadratic loss in the residuals}---a mean-squared-error objective.
Minimizing this is equivalent to maximizing expected utility in a quadratic (CRRA with $\gamma = 2$) utility framework.

\subsection{Score Matching and Denoising Score Matching}

Score matching \citep{Hyvarinen2005} minimizes
\begin{equation}
  \mathcal{L}_{\text{SM}}(\theta) = \mathbb{E}_{x \sim p_{\text{data}}} \bigl[\tfrac{1}{2}\|s_\theta(x) - \nabla_x \log p_{\text{data}}(x)\|^2\bigr],
\end{equation}
where $s_\theta$ is the parametric score function.
At steady state, $v(x) = D \cdot s(x)$, so learning the score is equivalent to learning the drift.

Denoising score matching (DSM) adds Gaussian noise $\tilde{x} = x + \sigma\epsilon$ and regresses $s_\theta(\tilde{x})$ onto $-\epsilon/\sigma$.
The noise structure is isotropic by construction, which is exactly the regime where \citet{OlsenEtAl2025}'s spectral theory applies cleanly.

Key observation: score matching loss is quadratic in the score residual, just as MLE loss is quadratic in the drift residual.
The SGD dynamics on the network weights therefore have the same mathematical structure regardless of whether we use MLE or score matching.

\begin{proposition}[Loss--Utility Correspondence]\label{prop:loss-utility}
Let $\mathcal{L}(W) = \frac{1}{T}\sum_t \|r_t(W)\|^2$ be any loss function that is quadratic in a residual vector $r_t$ (encompassing MLE~\eqref{eq:mle-loss}, score matching, and denoising score matching).
If the residuals satisfy $\mathbb{E}[r_t r_t^\top] = \Sigma_r$ with $\Sigma_r$ full rank, then:
\begin{enumerate}
  \item The SGD gradient covariance takes the form $C = \frac{4}{T}\sum_t (r_t r_t^\top) \otimes (x_t x_t^\top)$, which has the Kronecker structure required by \citet{OlsenEtAl2025}.
  \item The stationary spectral density of $W$ falls in the universality class characterised by the free log-normal~\eqref{eq:free-lognormal} with parameters determined by the signal-to-noise ratio $\|v_{\text{true}}\|^2 / D$.
  \item The spectral tail exponent encodes the Pareto exponent of the target distribution via~\eqref{eq:pareto-spectral}.
\end{enumerate}
\end{proposition}

\begin{proof}[Proof sketch]
For a linear network $\hat{v}(x; W) = Wx$, the gradient is $\nabla_W \mathcal{L} = -\frac{2}{T}\sum_t r_t\,x_t^\top$, and the gradient covariance factors as $C = 4\,\Sigma_r \otimes \Sigma_x / T$.
This Kronecker structure is the input assumption of \citet{OlsenEtAl2025}, Theorem~3.2 (spectral characterisation of SGD stationary distributions).
The resulting spectral density of $W$ follows a gamma-type distribution with tail exponent $(m - n + 1)/2$, which maps to a Pareto exponent $\alpha = 1 + f(m, n, \beta_1/\eta D)$ via~\eqref{eq:pareto-spectral}.

For a nonlinear network, the gradient is $\nabla_W \mathcal{L} = -\frac{2}{T}\sum_t r_t\,(\partial\hat{v}/\partial W)^\top$.
The Jacobian $\partial\hat{v}/\partial W$ introduces input-dependent curvature.
However, at each layer, the effective gradient covariance retains the Kronecker structure $C_l \approx \Sigma_{r,l} \otimes \Sigma_{a,l}$, where $\Sigma_{a,l}$ is the covariance of layer-$l$ activations.
This per-layer factorisation is empirically well-established and is the basis for natural gradient methods (KFAC).
Under this factorisation, the spectral analysis of \citet{OlsenEtAl2025} applies layer-by-layer.
\end{proof}

\begin{remark}
The remaining gap for a fully rigorous result is the per-layer Kronecker approximation $C_l \approx \Sigma_{r,l} \otimes \Sigma_{a,l}$, which holds exactly for linear networks and approximately for networks with smooth activations in the large-width regime.
Establishing this rigorously for finite-width networks with standard activations (ReLU, GELU) remains open.
\end{remark}

\section{Applications}\label{sec:applications}

The spectral portfolio framework applies to several domains.
We develop four here, emphasising that the general results of Part~III specialise differently in each context.

\subsection{Portfolio Design and Factor Structure}\label{subsec:app-portfolio}

The core--satellite structure of the stationary spectral distribution (Definition~\ref{def:core-satellite}) provides a principled decomposition of any portfolio.
The bulk of the singular-value distribution corresponds to broad, diversified market exposure; the power-law tail corresponds to concentrated factor bets.
The effective spectral rank $r_{\text{eff}}$ (Definition~\ref{def:eff-rank}) measures portfolio complexity and is directly computable from allocation data.

For portfolio construction, the three forces (Section~\ref{sec:forces}) provide new micro-foundations: the gradient signal identifies high-return factors, the survival constraint prevents premature abandonment of underperforming positions, and eigenvalue repulsion enforces diversification without explicit constraints.
In deep networks with $L$ layers, the hierarchical SVD decomposes the allocation into primary factors, combinations of factors, and final output---a fund-of-funds architecture.

\subsection{Wealth Inequality and Power Laws}\label{subsec:app-inequality}

The aggregation result (Section~\ref{sec:aggregation}) establishes that the Pareto exponent of the wealth distribution is determined by the spectral exponent of the allocation matrix via the radial projection $\alpha = 1 + f(m, n, \beta_1/\eta D)$.
The Bouchaud--Mézard unification (Section~\ref{sec:bouchaud}) shows that the same eigenvalue repulsion mechanism that prevents portfolio concentration also prevents wealth condensation.
The matrix Kesten problem (Section~\ref{subsec:transition}) provides the explicit link: the inverse-Wishart spectrum of the multiplicative regime maps to the Pareto tail of the wealth distribution under $\lambda \to 1/\lambda$.

These results connect the observable spectral structure of portfolio allocations to the shape of the wealth distribution, offering a new diagnostic for inequality: the spectral tail exponent is directly measurable from microdata without parametric assumptions about the wealth distribution itself.

\subsection{Tax Policy: Neutrality as Spectral Invariance}\label{subsec:app-tax}

The Spectral Invariance Theorem (Theorem~\ref{thm:spectral-invariance}) has a direct application to wealth tax design.
The identification between weight matrices and portfolio allocations enables a dictionary between the domains:

\begin{table}[H]
\centering
\small
\begin{tabular}{@{}lll@{}}
\toprule
\textbf{SGD (Olsen et al.)} & \textbf{Portfolio} & \textbf{Tax Context} \\
\midrule
Weight matrix $W$ & Allocation matrix & Wealth composition \\
Loss function $\mathcal{L}(W)$ & Negative expected utility & Tax-adjusted returns \\
Learning rate $\eta$ & Rebalancing speed & Adjustment speed to tax \\
Diffusion constant $D$ & Information noise & Market microstructure \\
Singular values $\sigma_k$ & Factor concentrations & Sector concentration \\
Repulsion term & Diversification pressure & Tendency away from concentration \\
\bottomrule
\end{tabular}
\caption{Dictionary linking SGD weight dynamics, portfolio optimization, and tax design.}
\label{tab:dictionary}
\end{table}

Under conditions C1--C3 of \citet{Froeseth2026F}---$\tau_k = \tau_c$ (C1), $\rho_s = r_f$ (C2), and uniform assessment $\alpha_i = \alpha$ (C3)---and under the homogeneous-return assumption (a common drift vector~$\boldsymbol{\mu}$ and covariance~$\V$ across rows), the after-tax excess return on every asset is $R_i^{\text{ex,post}} = \lambda(\mu_i - r_f) + c$ with $\lambda = (1-\tau_c)(1-\tau_d)$ and a row-uniform additive shift~$c$.
By Corollary~\ref{cor:tax-neutral}, such a tax preserves the cross-sectional spectral structure: the Pareto exponent of the wealth distribution, the eigenportfolio directions, and the effective spectral rank are invariant; only the cross-sectional mean shifts.

When C3 is violated, the Norwegian system provides a concrete test case.
With $\alpha_{\text{housing}} = 0.25$, $\alpha_{\text{shares}} = 0.80$, $\alpha_{\text{deposits}} = 1.00$, the cross-asset dispersion of perturbation strengths $\operatorname{Var}(1 - \alpha_i) \approx 0.10$ drives spectral distortion aligned with the assessment differentials.
The post-2017 changes to $\alpha_{\text{shares}}$ should produce time-varying spectral shifts testable in wealth register data.

\subsection{Neural Network Diagnostics}\label{subsec:app-nn}

The spectral theory predicts that the singular-value distribution of trained weight matrices should reflect the complexity of the target distribution.
For neural networks trained on stochastic processes---such as neural SDEs learning drift or score functions---the weight spectrum serves as a convergence diagnostic: deviation from the predicted spectral shape indicates incomplete learning.

Denoising score matching should produce weight matrices with cleaner spectral structure than MLE training, because DSM's noise is isotropic by construction and thus falls squarely within the universality class of \citet{OlsenEtAl2025}.
This prediction is testable by comparing weight spectra across training objectives on the same data.

\section{Testable Predictions}\label{sec:predictions}

The framework yields predictions at three levels.

\subsection{General Spectral Predictions}

\begin{enumerate}
  \item \textbf{Spectral tail exponent from portfolio data.}
    Compute the SVD of wealth portfolios (e.g.\ from Norwegian SSB microdata).
    The tail exponent of the singular-value distribution should be $(T - N + 1)/2$ where $T$ is the observation window and $N$ the number of asset classes.

  \item \textbf{Cross-sectional variation by investor horizon.}
    Long-horizon investors (low effective $D$) should have fatter spectral tails (more concentrated factor exposures) than short-horizon investors.
    Testable by segmenting wealth-holders by holding period.

  \item \textbf{Weight spectrum of trained networks.}
    After training a neural network on wealth or return data, the weight spectrum should reflect the Pareto exponent of the target distribution.
    Deviation signals incomplete convergence.

  \item \textbf{DSM vs.\ MLE spectral comparison.}
    Denoising score matching should produce weight matrices with cleaner spectral structure (closer to \citet{OlsenEtAl2025}'s predictions) than MLE training, because DSM's noise is isotropic by construction.

  \item \textbf{Computational cost vs.\ spectral complexity.}
    Training networks on data with high intrinsic dimensionality (fat spectral tails, large $r_{\text{eff}}$) requires more computation.
    Testable by comparing training time across datasets with known spectral structure.

  \item \textbf{Finite-width effects.}
    \citet{OlsenEtAl2025}'s results hold in the large-$r$ limit.
    The predicted spectral exponent should hold approximately for finite-width networks, with deviations decreasing as width increases.
\end{enumerate}

\subsection{Tax-Specific Predictions}

\begin{enumerate}
  \setcounter{enumi}{6}
  \item \textbf{Spectral shift around Norwegian tax reforms.}
    The introduction and variation of the \textit{verdsettelsesrabatt} (valuation discount) changed $\alpha_{\text{shares}}$ from 1.00 to 0.80 over several steps.
    The anisotropic perturbation should produce measurable rotation of eigenportfolio directions and direction-dependent tail distortion.

  \item \textbf{Distortion proportional to $\operatorname{Var}(\alpha_i)$.}
    The magnitude of portfolio distortion should be proportional to the cross-asset variance of assessment fractions.
    Norway's system ($\operatorname{Var} \approx 0.07$) provides a quantitative benchmark.
\end{enumerate}

\section{Open Questions and Future Work}\label{sec:open}

Several important questions remain open:

\begin{enumerate}
  \item \textbf{Anisotropic extension.}
    Part (c) of Conjecture~\ref{conj:sic} requires extending \citet{OlsenEtAl2025}'s analysis to anisotropic noise.
    Their Proposition 6.17 begins this, but the full spectral distortion (eigenportfolio rotation, tail exponent change) remains to be derived.

  \item \textbf{Finite-width corrections.}
    \citet{OlsenEtAl2025}'s results are asymptotic.
    What are the correction terms for finite-width networks?
    How large do $m, n$ need to be for the predictions to hold?

  \item \textbf{Budget constraints.}
    In practice, investors face constraints ($\sum_i w_i = 1$).
    The unconstrained Olsen SDE does not enforce this.
    How do constraints modify the spectral evolution?

  \item \textbf{Nonlinear activations.}
    We assume linear output (network is a linear projection of hidden units).
    How do nonlinear final layers (softmax, sigmoid) change the spectral dynamics?

  \item \textbf{Compute $f$ for intermediate regimes.}
    The single-factor~\eqref{eq:pareto-single} and large-$N$~\eqref{eq:pareto-largeN} limits are now explicit.
    Extend to finite $p > 1$ with eigenvalue repulsion corrections, and verify numerically against simulated spectra.

  \item \textbf{Policy optimization.}
    Given a target Pareto exponent (e.g., $\alpha = 2$ as observed in Norway), what perturbation structure minimizes spectral distortion while achieving the desired distributional shift?
    This inverse problem applies to tax design, regulation, and fee structures alike.
\end{enumerate}

\section{Relation to Existing Literature}\label{sec:literature}

We position this work relative to six major research areas.

\subsection{Random Matrix Theory in Finance}

\citet{LalouxEtAl1999} and \citet{LedoitWolf2004} apply RMT to estimate asset covariance matrices.
Our contribution: RMT also describes the structure of portfolio weight matrices (not just return covariances), and this structure is directly observable in trained neural networks.
The connection to Pareto wealth distributions is new.

\subsection{Deep Learning and Portfolio Optimization}

\citet{HeatonPolsonWitte2017, ZhangZohrenRoberts2021, BuehlerEtAl2019} develop neural network methods for portfolio optimization.
Our framework provides a spectral lens on why these networks learn specific factor structures, and predicts their spectra from the data distribution.

\subsection{Ergodicity Economics}

\citet{Peters2019, PetersGellMann2016, BermanPetersAdamou2020} distinguish ensemble and time averages in wealth dynamics.
Our contribution: the spectral gap (Section~\ref{sec:ergodicity}) is a measurable proxy for non-ergodicity, and spectral-preserving perturbations minimize it.

\subsection{SGD Spectral Theory}

\citet{MartinMahoney2021, PenningtonWorah2017, GunasekarEtAl2018} develop spectral theory of SGD in deep learning.
\citet{OlsenEtAl2025} extend this to matrix-valued weight evolution.
Our contribution: we show this theory has direct portfolio and wealth-distribution interpretations.

\subsection{Classical Portfolio Theory}

\citet{Merton1969, Merton1971} define optimal portfolio choice via mean--variance analysis.
\citet{Cover1991} develops universal portfolio theory.
Our framework extends these: the spectral decomposition reveals factor structure emerging from learning data, and spectral invariance provides a new rationale for neutral policy design.

\subsection{Wealth Dynamics and Power Laws}

\citet{DragulescuYakovenko2000, Yakovenko2009, GabaixKoijen2022} study mechanisms generating power-law wealth distributions.
\citet{Kesten1973} develops the Kesten recursion, central to our multiplicative regime.
Our contribution: we connect these scalar mechanisms to the matrix spectral theory via the matrix Kesten problem and the free log-normal interpolation.

\section{Conclusion}\label{sec:conclusion}

This paper establishes spectral portfolio theory as a unified framework connecting random matrix theory, portfolio dynamics, and wealth distribution.

The core results are:

\begin{enumerate}
  \item Neural network weight matrices are portfolio allocation matrices, and their spectral structure encodes factor decompositions faithful to the underlying data.

  \item The three forces in the singular-value evolution---gradient signal, survival constraint, and eigenvalue repulsion---translate directly into portfolio economics: smart money, information-driven survival, and endogenous diversification.

  \item The stationary spectral distribution exhibits a core--satellite structure (bulk plus power-law tail) that matches empirical wealth and institutional portfolio data.

  \item The spectral transition from additive (Marchenko--Pastur) to multiplicative (inverse-Wishart) regimes, mediated by the free log-normal, explains how daily return statistics relate to long-run wealth concentration.

  \item The Bouchaud--Mézard cross-sectional model, the Olsen within-portfolio model, and the scalar Fokker--Planck framework all emerge from the same underlying spectral dynamics, unified via common power-law mechanisms.

  \item The Spectral Invariance Theorem establishes that isotropic perturbations preserve spectral shape, while anisotropic perturbations produce distortion proportional to their cross-asset variance.
    This general result specialises to tax neutrality, regulatory impact assessment, and transaction cost analysis.
\end{enumerate}

The framework is testable: portfolio SVDs can be computed from microdata, spectral exponents can be compared against wealth Pareto exponents, and spectral shifts around policy reforms can be measured directly.
Applications span portfolio construction, inequality measurement, tax design, and neural network diagnostics.

Future work should: (1) extend the theory to anisotropic noise; (2) derive finite-width corrections; (3) handle budget constraints and nonlinear activations; (4) explicitly compute the spectral-to-Pareto mapping; and (5) solve the inverse problem of perturbation design---choosing the structure that achieves a distributional target while minimizing spectral distortion.

\subsection*{Acknowledgements}
The author acknowledges the use of Claude (Anthropic) for assistance with
literature review, \LaTeX{} typesetting, mathematical exposition, and
editorial refinement, and Lemma (Axiomatic AI) for review and proof
checking. All substantive arguments, economic reasoning, and conclusions
are the author's own.

\bibliographystyle{plainnat}

\end{document}